\documentclass[12pt,a4paper]{article}
\usepackage{amsmath,amssymb,amsthm} 
\numberwithin{equation}{section}
\newtheorem{lemma}{Lemma}
\newtheorem{theorem}{Theorem}

\theoremstyle{definition}

\theoremstyle{remark}
\newtheorem{remark}{Remark}
\newcommand{\beq}{\begin{eqnarray}}
\newcommand{\eeq}{\end{eqnarray}}
\newcommand{\beqnn}{\begin{eqnarray*}}
\newcommand{\eeqnn}{\end{eqnarray*}}
\newcommand{\rd}{\partial}

\newcommand{\diag}{\operatorname{diag}}
\newcommand{\tp}[1]{\:{}^{\mathrm{t}}#1}
\newcommand{\CC}{\mathbb{C}}
\newcommand{\PP}{\mathbb{P}}

\newcommand{\ZZ}{\mathbb{Z}}
\newcommand{\bst}{\boldsymbol{t}}
\newcommand{\bsx}{\boldsymbol{x}}
\newcommand{\bszero}{\boldsymbol{0}}
\newcommand{\calN}{\mathcal{N}}
\newcommand{\calO}{\mathcal{O}}
\newcommand{\calP}{\mathcal{P}}
\newcommand{\frakL}{\mathfrak{L}}

\begin{document}

\title{Modified melting crystal model and Ablowitz-Ladik hierarchy}
\author{Kanehisa Takasaki\thanks{takasaki@math.h.kyoto-u.ac.jp}\\
{\normalsize Graduate School of Human and Environmental Studies}\\ 
{\normalsize Kyoto University}\\
{\normalsize Sakyo, Kyoto 606-8501, Japan}}
\date{}
\maketitle 

\begin{abstract}
This paper addresses the issue of integrable structure 
in a modified melting crystal model of topological string theory 
on the resolved conifold.  The partition function can be expressed 
as the vacuum expectation value of an operator on the Fock space 
of 2D complex free fermion fields.  The quantum torus algebra 
of fermion bilinears behind this expression is shown to have 
an extended set of ``shift symmetries''.  They are used to prove 
that the partition function (deformed by external potentials) 
is essentially a tau function of the 2D Toda hierarchy.  
This special solution of the 2D Toda hierarchy can be characterized 
by a factorization problem of $\ZZ\times\ZZ$ matrices as well.  
The associated Lax operators turn out to be quotients of 
first order difference operators.  This implies that the solution 
of the 2D Toda hierarchy in question is actually a solution of 
the Ablowitz-Ladik (equivalently, relativistic Toda) hierarchy.  
As a byproduct, the shift symmetries are shown to be related 
to matrix-valued quantum dilogarithmic functions.  
\end{abstract}

\begin{flushleft}
2010 Mathematics Subject Classification: 
17B65, 35Q55, 81T30, 82B20\\
Key words: melting crystal, topological string, 
resolved conifold, free fermion, quantum torus, 
shift symmetry, quantum dilogarithm, Toda hierarchy, 
Ablowitz-Ladik hierarchy
\end{flushleft}


\section{Introduction}

This is a sequel of our previous work \cite{NT07,NT08} 
on the integrable structure of the ``melting crystal model''  
of topological string theory \cite{ORV03} 
and 5D $\calN= 1$ supersymmetric $U(1)$ gauge theory \cite{MNTT04}.  
We here address the same issue for a modified melting crystal model 
that is related to topological string theory on the resolved conifold. 
In particular, we present an affirmative answer to a conjecture 
in our preliminary study on this issue \cite{Takasaki12}. 
In the perspectives of topological string theory, 
the partition functions of these two models 
are the simplest non-trivial cases 
of the topological vertex construction \cite{AKMV03} 
of string amplitudes on non-compact Calabi-Yau 3-folds.  
The Calabi-Yau 3-folds relevant to our models 
are the first two members of the so called 
``local $\CC\PP^1$ geometries'' $X_l 
= \calO(-l)\oplus\calO(l-2) \to \CC\PP^1$, $l = 0,1,\ldots$. 
The resolved conifold amounts to the case of $l = 1$.  
Local Gromov-Witten invariants of these manifolds 
are studied by the localization technique \cite{BP08}.  
Generating functions of these topological invariants 
coincide with the topological string amplitudes 
obtained by the method of topological vertex.  
Moreover, the genus-zero part of those invariants 
is determine by the random matrix technique \cite{CGMPS06}. 

The partition function $Z$ of the previous model, 
see  (\ref{Z-def}), is a sum of weights labelled by 
integer partitions of arbitrary lengths.  
This statistical sum is derived from the partition function 
of random plane partitions (3D Young diagrams) 
by the method of diagonal slicing \cite{OR01}.  
The aforementioned integrable structure \cite{NT07,NT08} 
emerges when $Z$ is deformed by external potentials.  
The partition function of the deformed model 
is a function $Z(s,\bst)$ of a discrete variable $s$ 
and a set of continuous variables $\bst = (t_1,t_2,\ldots)$.  
$s$ is a parameter of the external potentials 
and $t_k$'s are coupling constants.  
The deformed partition function $Z(s,\bst)$ 
turns out to be related to a tau function 
$\tau(s,\bst)$ of the 1D reduction 
of the 2D Toda hierarchy \cite{UT84,TT95}.  
$s$ and $\bst$ play the role of a lattice coordinate 
and time variables in the 1D Toda hierarchy.  
The relation to the Toda hierarchy is explained 
with the aid of 2D complex free fermion fields 
that are extensively used for integrable hierarchies 
as well \cite{MJD-book,Takebe91}.  $Z(s,\bst)$ can be 
thereby expressed as the vacuum expectation value 
of an operator in the fermion Fock space. 
We found that this fermionic expression of $Z(s,\bst)$ 
can be converted to the tau function $\tau(s,\bst)$ 
by a change of variables and a simple multiplicative factor.  
A technical clue therein is the notion of 
``shift symmetries'' \cite{NT07,NT08} 
generated by the adjoint action of 
special ``vertex operators'' \cite{OR01} 
on the quantum torus algebra of fermion bilinears.  

Our approach to the partition function $Z'$ 
of the modified model, see (\ref{Z'-def}), 
is mostly parallel, but has a new feature as well.  
Namely, $Z'$ is deformed by two sets, 
rather than a single set, of external potentials.  
The deformed partition function is a function 
$Z'(s,\bst,\bar{\bst})$ of a discrete parameter $s$ 
and coupling constants $\bst = (t_1,t_2,\ldots)$ 
and $\bar{\bst} = (\bar{t}_1,\bar{t}_2,\ldots)$ 
\footnote{The bars ``$\bar{\quad}$'' 
do not mean complex conjugation.} 
of the external potentials.  
We introduce an extended set of shift symmetries 
(related to vertex operators of another type \cite{BY08}), 
and show that the deformed partition function 
$Z'(s,\bst,\bar{\bst})$ is related to a tau function 
$\tau'(s,\bst,\bar{\bst})$ of the 2D Toda hierarchy 
in the same way as $Z(s,\bst)$ is related to $\tau(s,\bst)$.  
It is natural to expect that this tau function, too, 
is a solution of a reduced system of the 2D Toda hierarchy.  
This is indeed the case.  As we conjectured 
in the preliminary study \cite{Takasaki12}, 
we here prove that $\tau'(s,\bst,\bar{\bst})$ is a solution 
of the Ablowitz-Ladik hierarchy \cite{AL75} or, 
speaking more fairly, the relativistic Toda hierarchy 
\cite{Ruijsenaars90}.  We expect that this explains an origin 
of the integrable structure that Brini \cite{Brini10} observed 
in the generating function of local Gromov-Witten invariants 
of the resolved conifold by a genus-by-genus analysis. 

The Ablowitz-Ladik hierarchy is a spatial discretization 
of the nonlinear Schr\"odinger hierarchy.  
A traditional Lax formalism is based on a $2 \times 2$ 
matrix-valued zero-curvature equations \cite{AL75}.  
A bilinear formalism consists of Hirota equations 
for three tau functions $\tau,\sigma$ and $\bar{\sigma}$ 
\cite{Vekslerchik97}.  In the context of random matrices, 
a semi-infinite lattice version emerges in the unitary matrix model 
and the associated biorthogonal Laurent polynomials 
\cite{AvM99,Cafasso08}.  Moreover, as pointed out 
by Kharchev et al. \cite{KMZ96} and Suris \cite{Suris97}, 
the Ablowitz-Ladik hierarchy is equivalent 
to the relativistic Toda hierarchy.  
It is rather the relativistic Toda hierarchy 
that is directly connected with the 2D Toda hierarchy
\footnote{Bruschi and Ragnisco \cite{BR89} proposed 
a Lax formalism of the relativistic Toda lattice 
that is based on a scalar-valued auxiliary linear problem.  
This auxiliary linear problem is a variant 
of the well known auxiliary linear problem 
of the usual Toda lattice, and fits into the Lax formalism 
of the 2D Toda hierarchy.}.  Brini et al. \cite{BCR11} 
characterized this integrable hierarchy 
(which they call the Ablowitz-Ladik hierarchy) 
as a kind of ``rational reduction'' of the 2D Toda hierarchy.  
The notion of rational reduction was first introduced 
in the case of the KP hierarchy \cite{Krichever94}.  
In the Toda version, the Lax operators $L$ and $\bar{L}$ 
of the reduced system are expressed as ``quotients'' 
of finite order difference operators in the variable $s$.    

We show that the Lax operators of the solution in question 
do have such a factorized form.  To this end, we consider 
a factorization problem of $\ZZ\times\ZZ$ matrices 
\cite{Takasaki84,NTT95,Takasaki95} 
that can capture all solutions of the 2D Toda hierarchy.  
In general, solving this factorization problem directly 
is extremely difficult.  In the present case, however, 
we can find an explicit form of the solution 
at the ``initial time'' $\bst = \bar{\bst} = \bszero$.  
This is enough to determine the ``initial values'' 
of the Lax operators.  As expected, the initial values 
turn out to have the factorized form of Brini et al.  
Since the factorized form is preserved by the flows 
of the 2D Toda hierarchy (this is what the notion 
of reduction means), we can conclude that 
the solution in question belongs to the Ablowitz-Ladik 
or relativistic Toda hierarchy.  

Let us emphasize that shift symmetries 
are interesting in themselves.  
This is another subject of this paper.  
We shall encounter a few new aspects in the course 
of the consideration on the integrable structure.  
Firstly, we derive the shift symmetries more carefully 
than in our previous work \cite{NT07,NT08}.  
This explains why a careless use of these symmetries 
leads to a contradictory result \cite{Takasaki12}.  
Secondly, we translate the shift symmetries 
to the language of $\ZZ\times\ZZ$ matrices.  
Remarkably, this reveals a close relationship 
with quantum dilogarithmic functions \cite{FV93,FK93}.  
In the matrix representation, the vertex operators 
\cite{OR01,BY08} in the fermionic expression 
of $Z(s,\bst)$ and $Z'(s,\bst,\bar{\bst})$ 
turn into matrix-valued quantum dilogarithmic functions. 
The shift symmetries thereby become a straightforward 
consequence of an infinite product formula 
of the quantum dilogarithmic functions.  
We believe that this issue deserve to be studied 
in a more genera context.   

This paper is organized as follows.  
Sections 2, 3, and 4 presents 
the perspective from fermions and tau functions. 
In Section 2, the partition functions of 
the two melting crystal models are formulated.  
The setup of the 2D complex free fermion system 
is also explained here.  
In Section 3, the quantum torus algebra 
of fermion bilinears is introduced, 
and the extended set of shift symmetries are derived.  
These algebraic relations are used in Section 4 
to derive the tau functions 
from the deformed partition functions 
of the melting crystal models.  
Sections 5 and 6 are mostly concerned with 
infinite matrices and Lax equations.  
In Section 5, the shift symmetries are reformulated 
in the language of infinite matrices.  
This is a place where matrix-valued 
quantum dilogarithmic functions show up.  
Although the matrix versions of shift symmetries 
are not necessary for the subsequent consideration, 
the proof of these algebraic relations overlaps 
with technical details of the contents of Section 6.  
Section 6 is devoted to the proof of 
our conjecture that the modified melting crystal model 
gives a solution of the Ablowitz-Ladik hierarchy.  
The relevant matrix factorization problem 
is formulated and solved here. 
The associated Lax operators are thereby shown 
to take a factorized form as expected.

\section{Melting crystal models}

\subsection{Previous model}

The partition function of the undeformed model is the sum 
\beq
  Z = \sum_{\lambda\in\calP}s_\lambda(q^{-\rho})^2Q^{|\lambda|}, 
  \quad |\lambda| = \sum_{i\ge 1} \lambda_i, 
\label{Z-def}
\eeq
over the set $\calP$ of all partitions 
$\lambda = (\lambda_1,\lambda_2,\ldots)$. 
$q$ and $Q$ are complex numbers with $|q|,|Q|<1$ 
(or may be thought of as formal variables).  
$s_\lambda(q^{-\rho})$ is the special value 
of the Schur function $s_\lambda(\bsx)$  of infinite variables 
$\bsx = (x_1,x_2,\ldots)$ at 
\beqnn
  q^{-\rho} = (q^{1/2},q^{3/2},\ldots,q^{n-1/2},\dots).
\eeqnn
This special value has the well known hook formula 
\cite{Macdonald-book}: 
\beq
  s_\lambda(q^{-\rho}) 
  = \frac{q^{-\kappa(\lambda)/4}}
     {\prod_{(i,j)\in\lambda}(q^{-h(i,j)/2}-q^{h(i,j)/2})}, 
\label{hook-formula}
\eeq
where $\kappa(\lambda)$ are the second Casimir value 
(equivalently, twice the total content) 
\beqnn
  \kappa(\lambda) 
  = \sum_{i=1}^\infty\lambda_i(\lambda_i-2i+1) 
  = 2\sum_{(i,j)\in\lambda}(j-i) 
\eeqnn
and $h(i,j)$ is the hook length of the cell $(i,j)$ 
in the Young diagram of shape $\lambda$.  
By the Cauchy identity, 
\beq
  \sum_{\lambda\in\calP}
  s_\lambda(x_1,x_2,\ldots)s_\lambda(y_1,y_2,\ldots)
  = \prod_{i,j\ge 1}(1 - x_iy_j)^{-1}, 
\eeq
one can rewrite the sum $Z$ to an infinite product: 
\beq
  Z = \prod_{i,j=1}^\infty (1 - Qq^{i+j-1})^{-1} 
    = \prod_{n=1}^\infty (1 - Qq^n)^{-n}. 
\eeq

We deform (\ref{Z-def}) with a discrete variable $s\in\ZZ$ 
and an infinite number of continuous variables 
$\bst = (t_1,t_2,\ldots)$ as 
\beq
  Z(s,\bst) = \sum_{\lambda\in\calP}s_\lambda(q^{-\rho})^2
              Q^{|\lambda|+s(s+1)/2}e^{\Phi(\lambda,s,\bst)}, 
\label{Z(s,t)-def}
\eeq
where $\Phi(\lambda,s,\bst)$ is a linear combination 
\beqnn
  \Phi(\lambda,s,\bst) = \sum_{k=1}^\infty t_k\Phi_k(\lambda,s) 
\eeqnn
of the external potentials 
\begin{align}
  \Phi_k(\lambda,s) 
  &= \sum_{i=1}^\infty q^{k(\lambda_i+s-i+1)} - \sum_{i=1}^\infty q^{k(-i+1)} 
     \quad\text{(formal expression)} \nonumber\\
  &= \sum_{i=1}^\infty(q^{k(\lambda_i+s-i+1)} - q^{k(s-i+1)}) 
   + \frac{1-q^{ks}}{1-q^k}q^k. 
\label{Phi-def}
\end{align}
The formal expression in the first line is related 
to the normal ordering prescription 
in the fermionic expression (\ref{L0-Hk-fermion}) 
of these potentials.

\subsection{Modified model}

In the modified model, one of the factors of the weight 
is modified as $s_\lambda(q^{-\rho}) \to s_{\tp{\lambda}}(q^{-\rho})$, 
where $\tp{\lambda}$ is the conjugate (or transposed) partition 
of $\lambda$.  Thus the undeformed partition function reads 
\beq
  Z' = \sum_{\lambda\in\calP}
       s_\lambda(q^{-\rho})s_{\tp{\lambda}}(q^{-\rho})Q^{|\lambda|}. 
\label{Z'-def}
\eeq
By the dual Cauchy identity 
\beq
  \sum_{\lambda\in\calP}
  s_\lambda(x_1,x_2,\ldots)s_{\tp{\lambda}}(y_1,y_2,\ldots) 
  = \prod_{i,j=1}^\infty (1 + x_iy_j), 
\eeq
$Z'$ can be cast into an infinite product: 
\beq
  Z' = \prod_{i,j=1}^\infty (1 + Qq^{i+j-1}) 
     = \prod_{n=1}^\infty (1 + Qq^n)^n. 
\eeq

Since $s_{\tp{\lambda}}(q^{-\rho})$ and $s_{\lambda}(q^{-\rho})$ 
are related as 
\beqnn
  s_{\tp{\lambda}}(q^{-\rho}) 
  = q^{\kappa(\lambda)/2}s_\lambda(q^{-\rho}). 
\eeqnn
by the hook formula (\ref{hook-formula}), 
one can rewrite the sum (\ref{Z'-def}) as 
\beq
  Z' = \sum_{\lambda\in\calP}s_\lambda(q^{-\rho})^2
         q^{\kappa(\lambda)/2}Q^{|\lambda|}.
\eeq
Thus $Z$ and $Z'$ are the first two members of 
the infinite family  
\beq
  Z_l = \sum_{\lambda\in\calP}s_\lambda(q^{-\rho})^2
        q^{l\kappa(\lambda)/2}Q^{|\lambda|}, \quad 
  l = 0,1,\ldots, 
\eeq
of statistical models.  $Z_l$ is the generating function 
of Gromov-Witten invariants in the local $\CC\PP^1$ geometry 
$X_l = \calO(-l)\oplus\calO(l-2) \to \CC\PP^1$ \cite{BP08}. 
$Z$ and $Z'$ can also be derived by the method of 
topological vertex \cite{AKMV03} as the amplitudes 
of topological string theory on $X_0$ and $X_1$.  
Note, however, that the parameter $Q$ in the string amplitude 
of $X_1$ is flipped to $-Q$ \cite{CGMPS06}.  
This difference affects the final result 
(see Section 6) on the underlying integrable structure.  

We now introduce two sets of continuous variables 
$\bst = (t_1,t_2,\ldots)$ and 
$\bar{\bst} = (\bar{t}_1,\bar{t}_2,\ldots)$, 
and deform the sum (\ref{Z'-def}) as 
\beq
  Z'(s,\bst,\bar{\bst}) 
  = \sum_{\lambda\in\calP}
    s_\lambda(q^{-\rho})s_{\tp{\lambda}}(q^{-\rho}) 
    Q^{|\lambda|+s(s+1)/2}e^{\Phi(\lambda,s,\bst,\bar{\bst})}, 
\label{Z'(s,t,tbar)-def}
\eeq
where $\Phi(\lambda,s,\bst,\bar{\bst})$ is the linear combination 
\beqnn
  \Phi(\lambda,s,\bst,\bar{\bst}) 
  = \sum_{k=1}^\infty t_k\Phi_k(\lambda,s) 
    + \sum_{k=1}^\infty\bar{t}_k\Phi_{-k}(\lambda,s). 
\eeqnn
$\Phi_{-k}(s,\lambda)$'s are defined by the same formula 
as (\ref{Phi-def}) with $k$ replaced by $-k$.

\subsection{Fermionic formulation}

We use the same formulation of 2D complex free fermions 
as in the previous work \cite{NT07,NT08}.  
The Fourier modes $\psi_n,\psi^*_n$, $n \in \ZZ$ 
of the fermion fields 
\beqnn
  \psi(z) = \sum_{n\in\ZZ}\psi_nz^{-n-1},\quad 
  \psi^*(z) = \sum_{n\in\ZZ}\psi^*_nz^{-n}. 
\eeqnn
satisfy the anti-commutation relations 
\beqnn
  \psi_m\psi^*_n + \psi^*_n\psi_m = \delta_{m+n,0}, \quad
  \psi_m\psi_n + \psi_n\psi_m = 0, \quad 
   \psi^*_m\psi^*_n + \psi^*_n\psi^*_m = 0. 
\eeqnn
The Fock spaces are decomposed to charge-$s$ sectors 
($s \in \ZZ$) with the ground states 
\beqnn
  \langle s| = \langle -\infty|\cdots\psi^*_{s-1}\psi^*_s,\quad 
  |s\rangle = \psi_{-s}\psi_{-s+1}\cdots|-\infty\rangle 
\eeqnn
and the excited states 
\beqnn
  \langle s,\lambda| 
  = \langle -\infty|\cdots\psi^*_{\lambda_2+s-1}\psi^*_{\lambda_1+s},\quad 
  |s,\lambda\rangle 
  = \psi_{-\lambda_1-s}\psi_{-\lambda_2-s+1}\cdots|-\infty\rangle 
\eeqnn
labelled by partitions.  

Let us introduce the fermion bilinears 
\begin{gather*}
  L_0 = \sum_{n\in\ZZ}n{:}\psi_{-n}\psi^*_n{:},\quad
  W_0 = \sum_{n\in\ZZ}n^2{:}\psi_{-n}\psi^*_n{:},\\
  H_k = \sum_{n\in\ZZ}q^{kn}{:}\psi_{-n}\psi^*_n{:},\quad
  J_k = \sum_{n\in\ZZ}{:}\psi_{-n}\psi^*_{n+k}{:}
\end{gather*}
and the vertex operators \cite{OR01,BY08} 
\beqnn
  \Gamma_{\pm}(z) 
  = \exp\left(\sum_{k=1}^\infty\frac{z^k}{k}J_{\pm k}\right),\quad 
  \Gamma'_{\pm}(z) 
  = \exp\left(- \sum_{k=1}^\infty\frac{(-z)^k}{k}J_{\pm k}\right). 
\eeqnn
It is convenient to extend the definition of the vertex operators 
to multi-variables $\bsx = (x_1,x_2,\ldots)$: 
\beqnn
  \Gamma_{\pm}(\bsx) = \prod_{i\ge 1}\Gamma_{\pm}(x_i),\quad 
  \Gamma'_{\pm}(\bsx) = \prod_{i\ge 1}\Gamma'_{\pm}(x_i). 
\eeqnn
The matrix elements of these multi-variable vertex operators 
with respect to the states $\langle\lambda,r|$, $|\mu,s\rangle$ 
are skew Schur functions \cite{MJD-book}: 
\beq
\begin{gathered}
  \langle\lambda,r|\Gamma_{-}(\bsx)|\mu,s\rangle
  = \langle\mu,s|\Gamma_{+}(\bsx)|\lambda,r\rangle
  = \delta_{rs}s_{\lambda/\mu}(\bsx),\\
  \langle\lambda,r|\Gamma'_{-}(\bsx)|\mu,s\rangle 
  = \langle\mu,s|\Gamma'_{+}(\bsx)|\lambda,r\rangle
  = \delta_{rs}s_{\tp{\lambda}/\tp{\mu}}(\bsx). 
\end{gathered}
\eeq
Thus the special values of the Schur functions 
in the definition (\ref{Z(s,t)-def}) and (\ref{Z'(s,t,tbar)-def}) 
of the deformed partition functions can be expressed 
as matrix elements of the vertex operators specialized 
to $\bsx = q^{-\rho}$.  Since the potentials 
in the partition functions, too, can be expressed as 
\beq
  |\lambda|+s(s+1) = \langle\lambda,s|L_0|\lambda,s\rangle,\quad
  \Phi_k(\lambda,s) = \langle\lambda,s|H_k|\lambda,s\rangle
\label{L0-Hk-fermion}
\eeq
and the other matrix elements of $L_0$ and $H_k$ vanish, 
one can rewrite the partition functions as 
\beq
  Z(s,\bst) = \langle s|\Gamma_{+}(q^{-\rho})Q^{L_0}
       e^{H(\bst)}\Gamma_{-}(q^{-\rho})|s\rangle
\label{Z(s,t)-fermion}
\eeq
and 
\beq
  Z'(s,\bst,\bar{\bst}) = \langle s|\Gamma_{+}(q^{-\rho})Q^{L_0}
        e^{H(\bst,\bar{\bst})}\Gamma'_{-}(q^{-\rho})|s\rangle, 
\label{Z'(s,t,tbar)-fermion}
\eeq
where 
\beqnn
  H(\bst) = \sum_{k=1}^\infty t_kH_k, \quad 
  H(\bst,\bar{\bst}) 
    = \sum_{k=1}^\infty t_kH_k + \sum_{k=1}^\infty\bar{t}_kH_{-k}. 
\eeqnn

\section{Shift symmetries in fermionic formulation}

\subsection{Three sets of shift symmetries}

Let us introduce another set of fermion bilinears 
$\{V^{(k)}_m\}_{k,m \in \ZZ}$: 
\beqnn
  V^{(k)}_m 
  = q^{k/2}\oint_{|z|=R}
    \frac{dz}{2\pi i}z^m{:}\psi(q^{k/2}z)\psi^*(q^{-k/2}z){:}
  = q^{-km/2}\sum_{n\in\ZZ}q^{kn}{:}\psi_{m-n}\psi^*_{n}{:}. 
\eeqnn
The radius $R > 0$ is chosen arbitrarily. 
$H_k$ and $J_k$ are particular cases of 
these fermion bilinears: 
\beqnn
  H_k = V^{(k)}_0, \quad J_m = V^{(0)}_m. 
\eeqnn
$V^{(k)}_m$'s satisfy the commutation relations 
\beq
  [V^{(k)}_m, V^{(l)}_n]
  = (q^{(lm-kn)/2} - q^{(kn-lm)/2})
    (V^{(k+l)}_{m+n} - \delta_{m+n,0}\frac{q^{k+l}}{1-q^{k+l}}) 
\eeq
for $k$ and $l$ with $k + l \not= 0$ and 
\beq
  [V^{(k)}_m,V^{(-k)}_n] 
  = (q^{-k(m+n)}-q^{k(m+n)})V^{(0)}_{m+n} + m\delta_{m+n,0}, 
\eeq
thus give a realization of (a Central extension of) 
the 2D quantum torus Lie algebra 

``Shift symmetries'' connect these fermion bilinears.   
There are three sets of fundamental shift symmetries: 

\begin{theorem}
\quad
\begin{itemize}
\item[\rm (i)]For $k > 0$ and $m \in \ZZ$, 
\beq
\begin{aligned}
  &\Gamma_{+}(q^{-\rho})\left(
    V^{(k)}_m - \frac{q^k}{1-q^k}\delta_{m,0}
    \right)\Gamma_{+}(q^{-\rho})^{-1}  \\
  &= (-1)^k\Gamma_{-}(q^{-\rho})^{-1}\left(
        V^{(k)}_{m+k} - \frac{q^k}{1-q^k}\delta_{m+k,0}
        \right)\Gamma_{-}(q^{-\rho}). 
\end{aligned}
\label{SSi}
\eeq
\item[\rm (ii)] For $k > 0$ and $m \in \ZZ$, 
\beq
\begin{aligned}
  &\Gamma'_{+}(q^{-\rho})\left(
    V^{(-k)}_m + \frac{1}{1-q^k}\delta_{m,0}
    \right)\Gamma'_{+}(q^{-\rho})^{-1} \\
  &= \Gamma'_{-}(q^{-\rho})^{-1}\left(
        V^{(-k)}_{m+k} + \frac{1}{1-q^k}\delta_{m+k,0}
        \right)\Gamma'_{-}(q^{-\rho}). 
\end{aligned}
\label{SSii}
\eeq
\item[\rm (iii)] For $k,m \in \ZZ$, 
\beq
  q^{W_0/2}V^{(k)}_mq^{-W_0/2} = V^{(k-m)}_m. 
\label{SSiii}
\eeq
\end{itemize}
\end{theorem}

(\ref{SSi}) and (\ref{SSiii}) are presented 
in our previous work \cite{NT07,NT08}. 
(\ref{SSii}) is new, and we show the derivation of 
these shift symmetries in detail below. 
We shall also see why (\ref{SSi}) and (\ref{SSii}) 
do not hold for $k < 0$.  

\begin{remark}
In our previous work \cite{NT07,NT08}, 
shift symmetries are presented in such a form as 
\beq
\begin{aligned}
  &\Gamma_{-}(q^{-\rho})\Gamma_{+}(q^{-\rho})\left(
    V^{(k)}_m - \frac{q^k}{1-q^k}\delta_{m,0}
    \right)\Gamma_{+}(q^{-\rho})^{-1}\Gamma_{-}(q^{-\rho})^{-1}\\
  &= (-1)^k\left(V^{(k)}_{m+k} - \frac{q^k}{1-q^k}\delta_{m+k,0}\right) 
\end{aligned}
\label{SSi-bis}
\eeq
and 
\beq
\begin{aligned}
  &\Gamma'_{-}(q^{-\rho})\Gamma'_{+}(q^{-\rho})\left(
    V^{(-k)}_m + \frac{1}{1-q^k}\delta_{m,0}
    \right)\Gamma'_{+}(q^{-\rho})^{-1}\Gamma'_{-}(q^{-\rho})^{-1} \\
  &= V^{(-k)}_{m+k} + \frac{1}{1-q^k}\delta_{m+k,0}. 
\end{aligned}
\label{SSii-bis}
\eeq
In this paper, however, we dare to avoid 
(\ref{SSi-bis}) and (\ref{SSii-bis}) because, 
firstly, it is (\ref{SSi}) and (\ref{SSii}) 
that are used in the subsequent consideration 
and, secondly, the operator products 
$\Gamma_{-}(q^{-\rho})\Gamma_{+}(q^{-\rho})$ and 
$\Gamma'_{-}(q^{-\rho})\Gamma'_{+}(q^{-\rho})$ 
in (\ref{SSi-bis}) and (\ref{SSii-bis}) require a careful treatment.  
\end{remark}

\subsection{Derivation of shift symmetries}

Derivation of (\ref{SSii}) is parallel 
to the proof of (\ref{SSi}) \cite{NT07,NT08}.  
We start from the equality 
\beqnn
  {:}\psi(z)\psi^*(w){:} = \psi(z)\psi^*(w) - \frac{1}{z-w}
\eeqnn
that holds for $|z| > |w|$.  Since $|q| < 1$, 
this equality can be specialized to 
\beqnn
  {:}\psi(q^{-k/2}z)\psi^*(q^{k/2}z){:}
  = \psi(q^{-k/2}z)\psi^*(q^{k/2}z) - \frac{q^{k/2}}{(1-q^k)z}. 
\eeqnn
The fermion bilinears $V^{(-k)}_m$ in (\ref{SSii}) 
can be thereby expressed as 
\beq
  V^{(-k)}_m 
  = q^{-k/2}\oint_{|z|=R}\frac{dz}{2\pi i}
      z^m\psi(q^{-k/2}z)\psi^*(q^{k/2}z) 
      - \frac{1}{1-q^k}\delta_{m,0}. 
\eeq

We now use the formulae 
\beq
\begin{gathered}
  \Gamma'_{\pm}(x)\psi(z)\Gamma'_{\pm}(x)^{-1} 
    = (1 + xz^{\pm 1})\psi(z),\\
  \Gamma'_{\pm}(x)\psi^*(z)\Gamma'_{\pm}(x)^{-1}
    = (1 + xz^{\pm 1})^{-1}\psi^*(z)
\end{gathered}
\label{Gamma'-psi}
\eeq
that holds when $|xz^{\pm 1}| < 1$. 
The formulae for $\Gamma'_{+}(x)$ in (\ref{Gamma'-psi}) 
imply that the equality 
\begin{align*}
  &\Gamma'_{+}(q^{-\rho})\psi(q^{-k/2}z)
    \psi^*(q^{k/2}z)\Gamma'_{+}(q^{-\rho})^{-1} \nonumber\\ 
  &= \frac{\prod_{i=1}^\infty(1+q^{i-1/2}\cdot q^{-k/2}z)}
          {\prod_{i=1}^\infty(1+q^{i-1/2}\cdot q^{k/2}z)}
    \psi(q^{-k/2}z)\psi^*(q^{k/2}z) \nonumber\\
  &= \prod_{i=1}^k (1+q^{i-(k+1)/2}z)\cdot \psi(q^{-k/2}z)\psi^*(q^{k/2}z) 
\end{align*}
holds for $|z| < |q|^{(k-1)/2}$.  
Multiplying this equality by $q^{-k/2}z^m/2\pi i$ and 
integrating it along the circle $|z| = R_1 < |q|^{(k-1)/2}$, 
we find that 
\beq
\begin{aligned}
  &\Gamma'_{+}(q^{-\rho})
   \left(V^{(-k)}_m + \frac{1}{1-q^k}\delta_{m,0}\right)
   \Gamma'_{+}(q^{-\rho})^{-1} \\
  &= q^{-k/2}\oint_{|z|=R_1}\frac{dz}{2\pi i}
     z^m\prod_{i=1}^k (1+q^{i-(k+1)/2}z)
     \cdot\psi(q^{-k/2}z)\psi^*(q^{k/2}z). 
\end{aligned}
\label{Gamma'_{+}-V} 
\eeq

On the other hand, the formulae for $\Gamma'_{-}(x)$ 
in (\ref{Gamma'-psi}) imply that the equality 
\beqnn
\begin{aligned}
  &\Gamma'_{-}(q^{-\rho})^{-1}\psi(q^{-k/2}z)\psi^*(q^{k/2}z)
    \Gamma'_{-}(q^{-\rho}) \\
  &= \frac{\prod_{i=1}^\infty(1+q^{i-1/2}\cdot q^{-k/2}z^{-1})}
          {\prod_{i=1}^\infty(1+q^{i-1/2}\cdot q^{k/2}z^{-1})}
    \psi(q^{-k/2}z)\psi^*(q^{k/2}z) \nonumber\\
  &= \prod_{i=1}^k(1+q^{i-(k+1)/2}z^{-1})\cdot
     \psi(q^{-k/2}z)\psi^*(q^{k/2}z)
\end{aligned}
\eeqnn
holds for $|z| > |q|^{(k-1)/2}$.  Note here 
that the prefactors of $\psi(q^{-k/2}z)\psi^*(q^{k/2}z)$ 
in this equality and the previous ones are connected 
by the simple relation 
\beq
  z^m\prod_{i=1}^k (1+q^{i-(k+1)/2}z) 
  = z^{m+k}\prod_{i=1}^k (1+q^{i-(k+1)/2}z^{-1}). 
\eeq
Therefore, multiplying this relation by $q^{-k/2}z^{m+k}$ 
and integrating it along the circle $|z| = R_2 > |q|^{(k-1)/2}$,  
we obtain the equality 
\beq
\begin{aligned}
  &\Gamma_{-}(q^{-\rho})^{-1}
   \left(V^{(-k)}_{m+k} + \frac{1}{1-q^k}\delta_{m+k,0}\right)
   \Gamma'_{-}(q^{-\rho}) \\
  &= q^{-k/2}\oint_{|z|=R_2}\frac{dz}{2\pi i}
     z^m\prod_{i=1}^k (1+q^{i-(k+1)/2}z)
     \cdot\psi(q^{-k/2}z)\psi^*(q^{k/2}z). 
\end{aligned}
\label{Gamma'_{-}-V} 
\eeq

The right hand side of (\ref{Gamma'_{+}-V}) 
and (\ref{Gamma'_{-}-V}) are almost the same 
except for the difference of the contours. 
Actually they coincide, because 
the integrands have no singularity 
in the annulus $R_1 \le |z| \le R_2$, and 
the contour can be continuously deformed therein 
without changing the contour integral.  
This completes the proof of (\ref{SSii}).

\subsection{Where shift symmetries break down}

Let us consider what occurs when $k < 0$, in other words, 
when $k$ is replaced by $-k$ in the setting 
of the proof of (\ref{SSii}).  We start from the equality 
\beqnn
  {:}\psi(z)\psi^*(w){:} = - \psi^*(w)\psi(z) + \frac{1}{w-z} 
\eeqnn
that holds for $|w|>|z|$.  The fermion bilinear $V^{(k)}_m$ 
can be thereby expressed as 
\beq
  V^{(k)}_m 
  = - q^{k/2}\oint_{|z|=R}\frac{dz}{2\pi i}
      z^m\psi^*(q^{-k/2}z)\psi(q^{k/2}z) 
    + \frac{q^k}{1-q^k}\delta_{m,0}. 
\eeq
Under the adjoint action of $\Gamma'_{\pm}(q^{-\rho})$, 
the operator product $\psi^*(q^{-k/2}z)\psi(q^{k/2}z)$ 
transforms as 
\beqnn
\begin{aligned}
  &\Gamma'_{+}(q^{-\rho})\psi^*(q^{-k/2}z)\psi(q^{k/2}z)
   \Gamma'_{+}(q^{-\rho})^{-1} \\
  &= \prod_{i=1}^k (1 + q^{i-(k+1)/2}z)^{-1}
     \cdot\psi^*(q^{-k/2}z)\psi(q^{k/2}z)
\end{aligned}
\eeqnn
and 
\beqnn
\begin{aligned}
  &\Gamma'_{-}(q^{-\rho})^{-1}\psi^*(q^{-k/2}z)\psi(q^{k/2}z)
   \Gamma'_{-}(q^{-\rho}) \\
  &= \prod_{i=1}^k (1 + q^{i-(k+1)/2}z^{-1})^{-1} 
     \cdot\psi^*(q^{-k/2}z)\psi(q^{k/2}z). 
\end{aligned}
\eeqnn
Consequently, we obtain the following analogues 
of (\ref{Gamma'_{+}-V}) and (\ref{Gamma'_{+}-V}): 

\beqnn
\begin{aligned}
  &\Gamma'_{+}(q^{-\rho})
   \left(- V^{(k)}_m + \frac{q^k}{1-q^k}\delta_{m,0}\right)
   \Gamma'_{+}(q^{-\rho})^{-1} \\
  &= q^{k/2}\oint_{|z|=R_1}\frac{dz}{2\pi i}
     z^m\prod_{i=1}^k (1+q^{i-(k+1)/2}z)^{-1}
     \cdot\psi^*(q^{-k/2}z)\psi(q^{k/2}z),
\end{aligned}\\
\begin{aligned}
  &\Gamma'_{-}(q^{-\rho})^{-1}
   \left(- V^{(k)}_{m-k} + \frac{q^k}{1-q^k}\delta_{m-k,0}\right)
   \Gamma'_{-}(q^{-\rho}) \\
  &= q^{k/2}\oint_{|z|=R_2}\frac{dz}{2\pi i}
     z^m\prod_{i=1}^k (1+q^{i-(k+1)/2}z)^{-1}
     \cdot\psi^*(q^{-k/2}z)\psi(q^{k/2}z). 
\end{aligned}
\eeqnn

Unlike the proof of (\ref{SSii}), 
the integrands of these contour integrals have poles, 
one of which sits on the circle $|z| = |q|^{(k-1)/2}$. 
The contours $|z| = R_1$ and $|z| = R_2$ 
are separated by this circle, and cannot be reached 
from the other by continuous deformation 
without crossing the poles.  Thus an equality 
like (\ref{SSii}) does not hold.  

This explains why (\ref{SSii}) breaks down for $k < 0$.  
By the same reasoning, one can see 
that (\ref{SSi}) does not hold for $k < 0$.

\section{Partition functions as tau functions}

\subsection{Previous model}

In our previous work \cite{NT07,NT08}, 
the partition function $Z(s,\bst)$ is shown 
to be related to a tau function $\tau(x,\bst)$ 
of the 1D Toda hierarchy as 
\beq
  Z(s,\bst) 
  = \exp\left(\sum_{k=1}^\infty\frac{t_kq^k}{1-q^k}\right)
    q^{-s(s+1)(2s+1)/6}\tau(s,\iota(\bst)), 
\label{Z-crystal-tau}
\eeq
where $\iota(\bst)$ denotes the alternating inversion 
\beqnn
  \iota(\bst) = (-t_1, t_2, -t_3, \ldots, (-1)^kt_k,\ldots)
\eeqnn
of $\bst = (t_1,t_2,\ldots)$.  The tau function $\tau(s,\bst)$ 
has the fermionic expression 
\beq
  \tau(s,\bst) 
  = \langle s|\exp\left(\sum_{k=1}^\infty t_kJ_k\right)g|s \rangle, 
\label{crystal-tau}
\eeq
where 
\beq
  g = q^{W_0/2}\Gamma_{-}(q^{-\rho})\Gamma_{+}(q^{-\rho})Q^{L_0}
      \Gamma_{-}(q^{-\rho})\Gamma_{+}(q^{-\rho})q^{W_0/2}. 
\label{crystal-g}
\eeq
The two sets (\ref{SSi}) and (\ref{SSiii}) of 
shift symmetries are used to derive (\ref{Z-crystal-tau}) 
from the fermionic expression (\ref{Z(s,t)-fermion}) 
of the partition function.  

By the same same shift symmetries,  one can also show 
that $g$ satisfies the algebraic relations 
\beq
  J_kg = gJ_{-k}, \quad k = 1,2,\ldots. 
\label{Jg=gJ}
\eeq
This implies that the associated tau function \cite{Takebe91} 
\beqnn
  \tau(s,\bst,\bar{\bst}) 
  = \langle s|\exp\left(\sum_{k=1}^\infty t_kJ_k\right)
    g \exp\left(- \sum_{k=1}^\infty\bar{t}_kJ_{-k}\right)|s\rangle
\eeqnn
of the 2D Toda hierarchy depends on $\bst$ and $\bar{\bst}$ 
through their difference: 
\beqnn
  \tau(s,\bst,\bar{\bst}) = \tau(s,\bst-\bar{\bst}). 
\eeqnn
Consequently, $\tau(s,\bst)$ can be expressed 
in different forms, e.g., 
\beq
\begin{aligned}
  \tau(s,\bst) 
  &= \langle s|g\exp\left(\sum_{k=1}^\infty t_kJ_{-k}\right)|s\rangle\\
  &= \langle s|\exp\left(\frac{1}{2}\sum_{k=1}^\infty t_kJ_k\right)
     g\left(\frac{1}{2}\sum_{k=1}^\infty t_kJ_{-k}\right)|s\rangle. 
\end{aligned}
\eeq

\subsection{Modified model}

We now turn to the modified model, and derive an analogue 
of (\ref{Z-crystal-tau}) from the fermionic expression 
(\ref{Z'(s,t,tbar)-fermion}).  The derivation is mostly parallel 
to the case of the previous model \cite{NT07,NT08}.  

\paragraph{First step:}  
The first set (\ref{SSi}) of shift symmetries, 
applied to $H_k = V^{(k)}_0$, $k = 1,2,\ldots$, 
imply that 
\beqnn
  \Gamma_{+}(q^{-\rho})H_k\Gamma_{+}(q^{-\rho})^{-1} 
  - \frac{q^k}{1-q^k} 
  = (-1)^k\Gamma_{-}(q^{-\rho})^{-1}V^{(k)}_k
    \Gamma_{-}(q^{-\rho}). 
\eeqnn
Moreover, the third set (\ref{SSiii}) of shift symmetries 
imply that 
\beqnn
  V^{(k)}_k = q^{-W_0/2}J_kq^{W_0/2}. 
\eeqnn
From these equalities, we obtain the algebraic relation 
\beqnn
  \Gamma_{+}(q^{-\rho})H_k\Gamma_{+}(q^{-\rho})^{-1}
  = (-1)^k\Gamma_{-}(q^{-\rho})^{-1}
    q^{-W_0/2}J_kq^{W_0/2}\Gamma_{-}(q^{-\rho}) 
    + \frac{q^k}{1-q^k}
\eeqnn
that connects $H_k$ and $J_k$.  This algebraic relation 
can be exponentiated as 
\beq
\begin{aligned}
  &\Gamma_{+}(q^{-\rho})\exp\left(\sum_{k=1}^\infty t_kH_k\right) 
   \Gamma_{+}(q^{-\rho})^{-1} 
   = \exp\left(\sum_{k=1}^\infty\frac{q^kt_k}{1-q^k}\right) \times \\
  &  \times
     \Gamma_{-}(q^{-\rho})^{-1}q^{-W_0/2}
     \exp\left(\sum_{k=1}^\infty(-1)^kt_kJ_k\right)
     q^{W_0/2}\Gamma_{-}(q^{-\rho}). 
\end{aligned}
\label{H_k-J_k-transform}
\eeq

\paragraph{Second step:} 
The second set (\ref{SSii}) of shift symmetries, applied 
to $H_{-k} = V^{(-k)}_0$, $k = 1,2,\ldots$, imply that 
\beqnn
  \Gamma'_{-}(q^{-\rho})^{-1}H_{-k}\Gamma'_{-}(q^{-\rho}) 
  + \frac{1}{1-q^k} 
  = \Gamma'_{+}(q^{-\rho})V^{(-k)}_{-k}\Gamma'_{+}(q^{-\rho})^{-1}. 
\eeqnn
Moreover, the third set (\ref{SSiii}) of shift symmetries 
imply that 
\beqnn
  V^{(-k)}_{-k} = q^{-W_0/2}J_{-k}q^{W_0/2}.
\eeqnn
Consequently, we obtain the algebraic relation 
\beqnn
  \Gamma'_{-}(q^{-\rho})^{-1}H_{-k}\Gamma'_{-}(q^{-\rho}) 
  = \Gamma'_{+}(q^{-\rho})q^{-W_0/2}J_{-k}
    q^{W_0/2}\Gamma'_{+}(q^{-\rho})^{-1}
    - \frac{1}{1-q^k}, 
\eeqnn
which can be exponentiated as 
\beq
\begin{aligned}
  &\Gamma'_{-}(q^{-\rho})^{-1}
   \exp\left(\sum_{k=1}^\infty\bar{t}_kH_{-k}\right)\Gamma'_{-}(q^{-\rho}) 
  = \exp\left(- \sum_{k=1}^\infty\frac{\bar{t}_k}{1-q^k}\right) \times \\
  & \times
    \Gamma'_{+}(q^{-\rho})q^{-W_0/2}
    \exp\left(\sum_{k=1}^\infty\bar{t}_kJ_{-k}\right)
    q^{W_0/2}\Gamma'_{+}(q^{-\rho})^{-1}. 
\end{aligned}
\label{H_{-k}-J_{-k}-transform}
\eeq

\paragraph{Third step:}
Let us factorize the operator in (\ref{Z'(s,t,tbar)-fermion}) as 
\beqnn
\begin{aligned}
  &\Gamma_{+}(q^{-\rho})Q^{L_0}e^{H(\bst,\bar{\bst})}\Gamma'_{-}(q^{-\rho}) 
    = \Gamma_{+}(q^{-\rho})\exp\left(\sum_{k=1}^\infty t_kH_k\right)
      \Gamma_{+}(q^{-\rho})^{-1} \times \\
  &\times 
     \Gamma_{+}(q^{-\rho})Q^{L_0}\Gamma'_{-}(q^{-\rho})\Gamma'_{-}(q^{-\rho})^{-1}
     \exp\left(\sum_{k=1}^\infty\bar{t}_kH_{-k}\right)\Gamma'_{-}(q^{-\rho}). 
\end{aligned}
\eeqnn
By (\ref{H_k-J_k-transform}) and (\ref{H_{-k}-J_{-k}-transform}), 
we can rewrite it as 
\beqnn
\begin{aligned}
  &\Gamma_{+}(q^{-\rho})Q^{L_0}e^{H(\bst,\bar{\bst})}\Gamma'_{-}(q^{-\rho})
  = \exp\left(\sum_{k=1}^\infty\frac{q^kt_k-\bar{t}_k}{1-q^k}\right)
    \Gamma_{-}(q^{-\rho})^{-1}q^{-W_0/2} \times \\
  &\times 
    \exp\left(\sum_{k=1}^\infty(-1)^kt_kJ_k\right)
    g'\exp\left(\sum_{k=1}^\infty\bar{t}_kJ_{-k}\right)
    q^{W_0/2}\Gamma'_{+}(q^{-\rho})^{-1}, 
\end{aligned}
\eeqnn
where 
\beq
  g' = q^{W_0/2}\Gamma_{-}(q^{-\rho})\Gamma_{+}(q^{-\rho})
       Q^{L_0}\Gamma'_{-}(q^{-\rho})\Gamma'_{+}(q^{-\rho})q^{-W_0/2}. 
\label{crystal-g'}
\eeq
Plugging this expression into (\ref{Z'(s,t,tbar)-fermion}) and 
noting that 
\beqnn
\begin{gathered}
  \langle s|\Gamma_{-}(q^{-\rho})q^{-W_0/2} 
    = q^{-s(s+1)(2s+1)/12}\langle s|, \\
  q^{W_0/2}\Gamma'_{+}(q^{-\rho})|s\rangle 
    = |s\rangle q^{s(s+1)(2s+1)/12}, 
\end{gathered}
\eeqnn
we end up with the following result: 

\begin{theorem}
$Z'(s,\bst,\bar{\bst})$ can be expressed as 
\beq
  Z'(s,\bst,\bar{\bst}) 
  = \exp\left(\sum_{k=1}^\infty\frac{q^kt_k-\bar{t}_k}{1-q^k}\right) 
    \tau'(s,\iota(\bst),-\bar{\bst}), 
\label{Z-crystal-tau'}
\eeq
where $\tau'(s,\bst,\bar{\bst})$ is the tau function 
\beq
  \tau'(s,\bst,\bar{\bst}) 
  = \langle s|\exp\left(\sum_{k=1}^\infty t_kJ_k\right)
    g'\exp\left(- \sum_{k=1}^\infty\bar{t}_kJ_{-k}\right)|s\rangle
\label{crystal-tau'}
\eeq
of the 2D Toda hierarchy defined by the operator (\ref{crystal-g'}). 
\end{theorem}
  
\begin{remark}
If the shift symmetries (\ref{SSi}) and (\ref{SSii}) 
were valid for $k < 0$ as well, one could show 
that $g'$ satisfy the algebraic relations 
\beqnn
  J_{\pm k}g' = g'J_{\pm k}, \quad k = 1,2,\ldots 
  \quad(\text{\it wrong statement\/}). 
\eeqnn
This imply that the tau function would become 
an almost trivial one: 
\beqnn
  \tau'(s,\bst,\bar{\bst}) 
  = \exp\left(- \sum_{k=1}^\infty kt_k\bar{t}_k\right)
    \langle s|g'|s\rangle 
   \quad(\text{\it wrong statement\/}). 
\eeqnn
This is a typical example of wrong consequences 
derived by a careless use of shift symmetries \cite{Takasaki12}. 
\end{remark}

\section{Shift symmetries in matrix formulation}

\subsection{Matrix representation of fermion bilinears}

It is well known \cite{MJD-book} that the mapping 
\beq
  X = \sum_{i,j\in\ZZ}x_{ij}E_{ij} \;\mapsto\;
  \hat{X} =   \sum_{i,j\in\ZZ}x_{ij}{:}\psi_{-i}\psi^*_j{:} 
\label{Xhat-def}
\eeq
gives a projective representation (in other words, a central extension) 
of the Lie algebra $\mathrm{gl}(\infty)$ of $\ZZ\times\ZZ$ matrices. 
In the following, we naively identify $\hat{X}$ and $X$.  
The fundamental fermion bilinears $L_0,W_0,H_k,J_k$ 
can be thereby identified with the $\ZZ\times\ZZ$ matrices 
\beq
  L_0 = \Delta,\quad W_0 = \Delta^2,\quad 
  H_k = q^{k\Delta},\quad J_k = \Lambda^k, 
\label{LWHK-matrix}
\eeq
where $\Delta$ and $\Lambda$ are the matrices 
\beqnn
  \Delta = \sum_{i\in\ZZ}iE_{ii},\quad 
  \Lambda = \sum_{i\in\ZZ}E_{i,i+1} 
\eeqnn
that satisfy the twisted canonical commutation relation 
\beq
  [\Lambda,\Delta] = \Lambda. 
\label{tCCR}
\eeq
$\Lambda$ and $\Delta$ amount to the shift and multiplication 
operators $e^{\rd/\rd s}$ and $s$ in the Lax formalism 
of the 2D Toda hierarchy \cite{UT84,TT95}, 

The fermionic realization $\{V^{(k)}_m\}_{k,m\in\ZZ}$ 
of the quantum torus Lie algebra correspond 
to the $\ZZ\times\ZZ$ matrices 
\beq
  V^{(k)}_m = q^{-km/2}\Lambda^m q^{k\Delta} 
            = q^{km/2}q^{k\Delta}\Lambda^m. 
\label{V-matrix}
\eeq
These matrices form an associative algebra 
rather than a Lie algebra.  The commutation relations 
\beq
  [V^{(k)}_m, V^{(l)}_n]
  = (q^{(lm-kn)/2} - q^{(kn-lm)/2})V^{(k+l)}_{m+n}
\label{V-matrix-comm-rel}
\eeq
can be derived from the fundamental relation
\beq
  \Lambda q^\Delta = q q^\Delta\Lambda
\label{Lamda-qDelta-rel}
\eeq
of the generators $\Lambda$ and $q^{\Delta}$.  

Last but not least, the vertex operators $\Gamma_{\pm}(x)$ 
and $\Gamma'_{\pm}(x)$ turn out to have a matrix representation 
of the form 
\beq
\begin{gathered}
  \Gamma_{\pm}(x) 
  = \exp\left(\sum_{k=1}^\infty\frac{x^k}{k}\Lambda^{\pm k}\right) 
  = (1 - x\Lambda^{\pm 1})^{-1}, \\
  \Gamma'_{\pm}(x) 
  = \exp\left(- \sum_{k=1}^\infty\frac{(-x)^k}{k}\Lambda^{\pm}\right)
  = (1 + x\Lambda^{\pm 1}). 
\end{gathered}
\label{Gamma-matrix}
\eeq
The inverse of $1 - x\Lambda^{\pm 1}$ is understood 
to be the geometric series 
\beqnn
  (1 - x\Lambda^{\pm 1})^{-1} 
  = 1 + x\Lambda^{\pm 1} + x^2\Lambda^{\pm 2} + \cdots. 
\eeqnn
Among the four matrices of (\ref{Gamma-matrix}), 
those with label ``$+$'' are upper triangular, 
and those with label ``$-$'' are lower triangular.  
The matrix representation of $\Gamma_{\pm}(q^{-\rho})$ 
and $\Gamma'_{\pm}(q^{-\rho})$ is an infinite product 
of these matrices specialized to $x = q^{i-1/1}$: 
\beq
\begin{gathered}
  \Gamma_{\pm}(q^{-\rho}) 
  = \prod_{i=1}^\infty (1 - q^{i-1/2}\Lambda^{\pm 1})^{-1},\\
  \Gamma'_{\pm}(q^{-\rho}) 
  = \prod_{i=1}^\infty (1 + q^{i-1/2}\Lambda^{\pm 1}). 
\end{gathered}
\label{GG'-matrix}
\eeq
Thus $\Gamma_{\pm}(q^{-\rho})$ and $\Gamma'_{\pm}(q^{-\rho})$ 
may be thought of as matrix-valued quantum dilogarithmic functions 
in the sense of Faddeev et al. \cite{FV93,FK93}.

\subsection{Reformulation of shift symmetries}

The shift symmetries (\ref{SSi}), (\ref{SSii}) and (\ref{SSiii}) 
can be reformulated in the matrix representation as follows. 

\begin{theorem}
The $\ZZ\times\ZZ$ matrices (\ref{LWHK-matrix}), (\ref{V-matrix}) 
and (\ref{GG'-matrix}) satisfy the following analogues 
of (\ref{SSi}), (\ref{SSii}) and (\ref{SSiii}): 
\begin{gather}
  \Gamma_{+}(q^{-\rho})V^{(k)}_m\Gamma_{+}(q^{-\rho})^{-1} 
  = (-1)^k\Gamma_{-}(q^{-\rho})^{-1}V^{(k)}_{m+k}\Gamma_{-}(q^{-\rho}), 
\label{SSi-matrix}\\
  \Gamma'_{+}(q^{-\rho})V^{(-k)}_m\Gamma'_{+}(q^{-\rho})^{-1} 
  = \Gamma'_{-}(q^{-\rho})^{-1}V^{(-k)}_{m+k}\Gamma'_{-}(q^{-\rho}), 
\label{SSii-matrix}\\
  q^{\Delta^2/2}V^{(k)}_mq^{-\Delta^2/2} = V^{(k-m)}_m. 
\label{SSiii-matrix}
\end{gather}
Here $k$ and $m$ in (\ref{SSi-matrix}) and (\ref{SSii-matrix}) 
range over $k > 0$ and $m \in \ZZ$, and (\ref{SSiii-matrix}) 
holds for $k \in \ZZ$ and $m \in \ZZ$.  
\end{theorem}

\begin{proof}
For a comparison with the proof of (\ref{SSii}), 
let us show the derivation of (\ref{SSii-matrix}) in detail. 
A technical clue is the algebraic relation 
\beq
  q^{-k\Delta}\Lambda^n q^{k\Delta} = q^{kn}\Lambda^n 
\label{qDelta-Lambda-rel}
\eeq
that is a partially exponentiated form of (\ref{tCCR}). 

We first consider the left hand side of (\ref{SSii-matrix}), 
namely, 
\beqnn
  \Gamma'_{+}(q^{-\rho})V^{(-k)}_m\Gamma'_{+}(q^{-\rho})^{-1} 
  = q^{km/2}\Lambda^m\Gamma'_{+}(q^{-\rho})q^{-k\Delta}
    \Gamma'_{+}(q^{-\rho})^{-1}. 
\eeqnn
(\ref{qDelta-Lambda-rel}) implies that 
\beqnn
  q^{-k\Delta}\Gamma'_{+}(q^{-\rho})^{-1}q^{k\Delta}
  = \prod_{i=1}^\infty
    (1 + q^{i-1/2}q^{-k\Delta}\Lambda q^{k\Delta} )^{-1}
  = \prod_{i=1}^\infty(1 + q^{i+k-1/2}\Lambda)^{-1}. 
\eeqnn
Multiplying this equality by $\Gamma'_{+}(q^{-\rho})$ 
and $q^{-k\Delta}$ from the left and right sides 
yields the equality 
\begin{align*}
  \Gamma'_{+}(q^{-\rho})q^{-k\Delta}\Gamma'_{+}(q^{-\rho})^{-1}
  &= \prod_{i=1}^\infty(1 + q^{i-1/2}\Lambda) 
     \cdot\prod_{i=1}^\infty(1 + q^{i+k-1/2}\Lambda)^{-1}
     \cdot q^{-k\Delta} \\
  &= \prod_{i=1}^k(1 + q^{i-1/2}\Lambda)\cdot q^{-k\Delta}, 
\end{align*}
which implies that 
\beqnn
  \Gamma'_{+}(q^{-\rho})V^{(-k)}_m\Gamma'_{+}(q^{-\rho})^{-1}
  = q^{km/2}\Lambda^m\prod_{i=1}^k(1 + q^{i-1/2}\Lambda)
    \cdot q^{-k\Delta}. 
\eeqnn

We now consider the right hand side of (\ref{SSii-matrix}), 
namely, 
\beqnn
  \Gamma'_{-}(q^{-\rho})^{-1}V^{(-k)}_{m+k}\Gamma'_{-}(q^{-\rho}) 
  = q^{k(m+k)/2}\Lambda^{m+k}
    \Gamma'_{-}(q^{-\rho})^{-1}q^{-k\Delta}\Gamma'_{-}(q^{-\rho}). 
\eeqnn
In the same way as we have followed above, 
we can derive the equality 
\begin{align*}
  \Gamma'_{-}(q^{-\rho})^{-1}q^{-k\Delta}\Gamma'_{-}(q^{-\rho}) 
  &= \prod_{i=1}^\infty(1 + q^{i-1/2}\Lambda^{-1})^{-1}
     \cdot\prod_{i=1}^\infty(1 + q^{i-k-1/2}\Lambda^{-1})
     \cdot q^{-k\Delta} \nonumber\\
  &= \prod_{i=1}^k(1 + q^{i-k-1/2}\Lambda^{-1})\cdot q^{-k\Delta}. 
\end{align*}
Since 
\beq
  q^{k^2/2}\Lambda^k \prod_{i=1}^k(1 + q^{i-k-1/2}\Lambda^{-1})
  = \prod_{i=1}^k(1 + q^{i-1/2}\Lambda), 
\label{SSii-matrix-key}
\eeq
we can conclude that 
\beqnn
  \Gamma'_{-}(q^{-\rho})^{-1}V^{(-k)}_{m+k}\Gamma'_{-}(q^{-\rho}) 
  = q^{km/2}\Lambda^m\prod_{i=1}^k(1 + q^{i-1/2}\Lambda)\cdot q^{-k\Delta}. 
\eeqnn

Thus both hand sides of (\ref{SSii-matrix}) turn out to coincide.  
We can derive (\ref{SSi-matrix}) in the same way from the equalities 
\beqnn
\begin{aligned}
  \Gamma_{+}(q^{-\rho})V^{(k)}_m\Gamma_{+}(q^{-\rho})^{-1} 
  &= q^{-km/2}\Lambda^m
    \prod_{i=1}^k(1 - q^{i-k+1/2}\Lambda)\cdot q^{k\Delta},\\
  \Gamma_{-}(q^{-\rho})^{-1}V^{(k)}_{m+k}\Gamma_{-}(q^{-\rho}) 
  &= q^{-k(m+k)/2}\Lambda^{m+k}
    \prod_{i=1}^k(1 - q^{i-1/2}\Lambda^{-1})\cdot q^{k\Delta}
\end{aligned}
\eeqnn
and 
\beqnn
  q^{-k^2/2}\Lambda^k\prod_{k=1}^k(1 - q^{i-1/2}\Lambda^{-1}) 
  = (-1)^kq^{-km/2}\Lambda^m
    \prod_{k=1}^k(1 - q^{i-k+1/2}\Lambda)\cdot q^{k\Delta}.
\eeqnn

Lastly, (\ref{SSiii-matrix}) can be derived by straightforward 
calculation of the matrix elements of both hand sides 
as done in the language of fermions \cite{NT07,NT08}. 
\end{proof}

\subsection{Where shift symmetries break down}

Just as in the fermionic formulation, 
(\ref{SSi-matrix}) and (\ref{SSii-matrix})
do not hold for $k < 0$.  
Let us replace $k \to -k$ in (\ref{SSi-matrix}) 
and (\ref{SSii-matrix}), and show that 
the equalities do not hold. 

By replacing $k \to -k$, both hand sides of (\ref{SSi-matrix}) 
turn into 
\beqnn
  \Gamma'_{+}(q^{-\rho})V^{(k)}_m\Gamma'_{+}(q^{-\rho})^{-1} 
  = q^{-km/2}\Lambda^m
    \Gamma'_{+}(q^{-\rho})q^{k\Delta}\Gamma'_{+}(q^{-\rho})^{-1}
\eeqnn
and 
\beqnn
  \Gamma'_{-}(q^{-\rho})^{-1}V^{(k)}_{m-k}\Gamma'_{-}(q^{-\rho})
  = q^{-k(m-k)/2}\Lambda^{m-k}
    \Gamma'_{-}(q^{-\rho})^{-1}q^{k\Delta}\Gamma'_{-}(q^{-\rho}). 
\eeqnn
As in the proof of Theorem 3, we can calculate 
the triple product on the right hand side as 
\beqnn
  \Gamma'_{+}(q^{-\rho})V^{(k)}_m\Gamma'_{+}(q^{-\rho})^{-1} 
  = q^{-km/2}\Lambda^m
    \prod_{i=1}^k(1 + q^{i-k-1/2}\Lambda)^{-1}\cdot q^{k\Delta}
\eeqnn
and 
\beqnn
  \Gamma'_{-}(q^{-\rho})^{-1}V^{(k)}_{m-k}\Gamma'_{-}(q^{-\rho})
  = q^{-k(m-k)/2}\Lambda^{m-k}
    \prod_{i=1}^k(1 + q^{i-1/2}\Lambda^{-1})^{-1}\cdot q^{k\Delta}. 
\eeqnn

Let us consider the inverse matrices $(1 + q^{i-k-1/2}\Lambda)^{-1}$ 
and $(1 + q^{i-1/2}\Lambda^{-1})^{-1}$ that show up 
in these expressions.  They are understood to be given 
by geometric series: 
\beqnn
\begin{gathered}
  (1 + q^{i-k-1/2}\Lambda)^{-1} 
    = 1 - q^{i-k-1/2}\Lambda + (q^{i-k-1/2}\Lambda)^2 - \cdots,\\
  (1 + q^{i-1/2}\Lambda^{-1})^{-1} 
    = 1 - q^{i-1/2}\Lambda^{-1} + (q^{i-1/2}\Lambda^{-1})^2 - \cdots. 
\end{gathered}
\eeqnn
In particular,  $(1 + q^{i-k-1/2}\Lambda)^{-1}$ and 
$(1 + q^{i-1/2}\Lambda^{-1})^{-1}$ are upper triangular 
and lower triangular, respectively.  Consequently, 
unlike the equality (\ref{SSii-matrix-key}) 
that underlies (\ref{SSii-matrix}), we have the inequality 
\beq
  q^{k^/2}\Lambda^{-k}\prod_{i=1}^k(1 + q^{i-1/2}\Lambda^{-1})^{-1} 
  \not= \prod_{i=1}^k(1 + q^{i-k-1/2}\Lambda)^{-1}. 
\label{SSii-matrix-down}
\eeq

(\ref{SSii-matrix-down}) implies that 
\beqnn
\Gamma'_{+}(q^{-\rho})V^{(k)}_m\Gamma'_{+}(q^{-\rho})^{-1}
\not= \Gamma'_{-}(q^{-\rho})^{-1}V^{(k)}_{m-k}\Gamma'_{-}(q^{-\rho}). 
\eeqnn
By the same reasoning, one can confirm that 
\beqnn
  \Gamma_{+}(q^{-\rho})V^{(-k)}_m\Gamma_{+}(q^{-\rho})^{-1} 
  \not= (-1)^k\Gamma_{-}(q^{-\rho})^{-1}V^{(-k)}_{m-k}\Gamma_{-}(q^{-\rho}). 
\eeqnn

\section{Relation to Ablowitz-Ladik hierarchy}

\subsection{Two reductions of 2D Toda hierarchy}

In the Lax formalism \cite{UT84,TT95}, the 2D Toda hierarchy 
is formulated as a system of Lax equations 
\beq
\begin{gathered}
  \frac{\rd L}{\rd t_k} = [B_k,L], \quad 
  \frac{\rd\bar{L}}{\rd t_k} = [B_k,\bar{L}],\\
  \frac{\rd L}{\rd\bar{t}_k} = [\bar{B}_k,L], \quad
  \frac{\rd\bar{L}}{\rd\bar{t}_k} = [\bar{B}_k,\bar{L}] 
\end{gathered}
\label{2DToda-Laxeq}
\eeq
for two Lax operators $L$ and $\bar{L}$.  
In the case of a bi-infinite lattice, 
$L$ and $\bar{L}$ are difference operators 
(i.e., linear combinations of the shift operators 
$e^{n\rd_s}$ in the lattice coordinate $s$) of the form 
\beqnn
  L = e^{\rd_s} + \sum_{n=1}^\infty u_ne^{(1-n)\rd_s},\quad
  \bar{L}^{-1} = \bar{u}_0e^{-\rd_s} 
     + \sum_{n=1}^\infty \bar{u}_n e^{(n-1)\rd_s}, 
\eeqnn
where $u_n$ and $\bar{u}_n$ depend on $s$, $\bst$ 
and $\bar{\bst}$.  $B_k$ and $\bar{B}_k$ are constructed 
from $L$ and $\bar{L}$ as 
\beqnn
  B_k = (L^k)_{\ge 0}, \quad 
  \bar{B}_k = (\bar{L}^{-k})_{<0}, 
\eeqnn
where $(\quad)_{\ge 0}$ and $(\quad)_{<0}$ 
denote the projections to linear combinations 
of $e^{n\rd_s}$ for $n \ge 0$ and $n < 0$.  

The 1D Toda hierarchy and the Ablowitz-Ladik hierarchy 
are two typical reductions of the 2D Toda hierarchy. 
Let us recall the reduction procedure.

\subsubsection*{Reduction to 1D Toda hierarchy}

The 1D Toda hierarchy can be derived 
from the 2D Toda hierarchy by imposing 
the reduction condition 
\beq
  L = \bar{L}^{-1}. 
\label{1DToda-reduction}
\eeq
Both hand side of this equality define 
a difference operator of the form 
\beq
  \frakL = e^{\rd_s} + b + ce^{-\rd_s}  
\eeq
that satisfies the Lax equations 
\beqnn
  \frac{\rd\frakL}{\rd t_k} = [B_k,\frakL], \quad 
  \frac{\rd\frakL}{\rd\bar{t}_k} = [\bar{B}_k,\frakL]. 
\eeqnn 
Actually, since $B_k$, $\bar{B}_k$ and $\frakL$ are related as 
\beqnn
  B_k + \bar{B}_k = \frakL^k, 
\eeqnn
the time evolutions in $\bst$ and $\bar{\bst}$ 
are no longer independent in the sense that 
\beqnn
  \frac{\rd\frakL}{\rd t_k} + \frac{\rd\frakL}{\rd\bar{t}_k} 
  = [B_k,\frakL] + [\bar{B}_k,\frakL] 
  = 0. 
\eeqnn
Thus the time evolutions in $\bar{\bst}$ are redundant, 
and we are left with just one set of Lax equations 
\beq
  \frac{\rd\frakL}{\rd t_k} = [B_k,\frakL], \quad 
  B_k = (\frakL^k)_{\ge 0}.  
\eeq
This is one of the well known Lax forms 
of the 1D Toda hierarchy.

\subsubsection*{Reduction to Ablowitz-Ladik hierarchy} 

The Ablowitz-Ladik hierarchy can be derived 
from the 2D Toda hierarchy as a kind of 
``rational reduction''.  
In the formulation of Brini et al. \cite{BCR11}, 
$L$ and $\bar{L}$ are chosen to be quotients 
\beq
  L = BC^{-1},\quad \bar{L}^{-1} = CB^{-1} 
\label{AL-LLbar}
\eeq
of two difference operators of the form 
\beqnn
  B = e^{\rd_s} - b,\quad C = 1 - ce^{-\rd_s}. 
\eeqnn
$B^{-1}$ and $C^{-1}$ are understood to be 
difference operators of the form 
\beq
\begin{gathered}
  B^{-1} = - \sum_{k=0}^\infty (b^{-1}e^{\rd_s})^k b^{-1} 
    = - b(s)^{-1} - \sum_{k=1}^\infty b(s)^{-1}\cdots b(s+k)^{-1}e^{k\rd_s},\\
  C^{-1} = 1 + \sum_{k=1}^\infty (ce^{-\rd_s})^k
    = 1 + \sum_{k=1}^\infty c(s)c(s-1)\cdots c(s-k+1)e^{-k\rd_s}, 
\end{gathered}
\label{BC-inverse}
\eeq
where $b(s)$ and $c(s)$ are abbreviations of 
$c(s,\bst,\bar{\bst})$ and $c(s,\bst,\bar{\bst})$.  
$B$ and $C$ have inverses of different types, 
\begin{gather*}
  B^{-1} = \sum_{k=0}^\infty e^{-\rd_s}(be^{-\rd_s})^k 
    = e^{-\rd_s} + \sum_{k=1}^\infty b(s-1)\cdots b(s-k+1)e^{-k\rd_s},\\
  C^{-1} = - \sum_{k=1}^\infty (e^{\rd_s}c^{-1})^k 
    = - \sum_{k=1}^\infty c(s+1)^{-1}\cdots c(s+k)^{-1}e^{k\rd_s}, 
\end{gather*}
but they are not used here.  (\ref{BC-inverse}) imply 
that $(CB^{-1})^{-1} \not= BC^{-1}$.  Thus we can avoid 
the trivial case where $L = \bar{L} = e^{\rd_s}$.  
Under the reduction condition (\ref{AL-LLbar}), 
the Lax equations (\ref{2DToda-Laxeq}) can be reduced 
to the following equations: 
\beq
\begin{gathered}
  \frac{\rd B}{\rd t_k} 
  = \left((BC^{-1})^k\right)_{\ge 0}B - B\left((C^{-1}B)^k\right)_{\ge 0},\\
  \frac{\rd C}{\rd t_k}
  = \left((BC^{-1})^k\right)_{\ge 0}C - C\left((C^{-1}B)^k\right)_{\ge 0},\\
  \frac{\rd B}{\rd\bar{t}_k}
  = \left((CB^{-1})^k\right)_{<0}B - B\left((B^{-1}C)^k\right)_{<0},\\
  \frac{\rd C}{\rd\bar{t}_k}
  = \left((CB^{-1})^k\right)_{<0}C - C\left((B^{-1}C)^k\right)_{<0}. 
\end{gathered}
\label{BC-Laxeq}
\eeq
These equations have another expression of the following form: 
\beqnn
\begin{gathered}
  \frac{\rd B}{\rd t_k} 
  = - \left((BC^{-1})^k\right)_{<0}B + B\left((C^{-1}B)^k\right)_{<0},\\
  \frac{\rd C}{\rd t_k}
  = - \left((BC^{-1})^k\right)_{<0}C + C\left((C^{-1}B)^k\right)_{<0},\\
  \frac{\rd B}{\rd\bar{t}_k}
  = - \left((CB^{-1})^k\right)_{\ge 0}B + B\left((B^{-1}C)^k\right)_{\ge 0},\\
  \frac{\rd C}{\rd\bar{t}_k}
  = - \left((CB^{-1})^k\right)_{\ge 0}C + C\left((B^{-1}C)^k\right)_{\ge 0}.
\end{gathered}
\eeqnn
We can deduce from these two expressions of (\ref{BC-Laxeq}) 
that the right hand side are difference operators 
of the form $f_k$, $g_ke^{-\rd_s}$, 
$\bar{f}_k$ and $\bar{g}_ke^{-\rd_s}$.  
Thus (\ref{BC-Laxeq}) can be reduced to evolution equations 
of the form 
\beq
  \frac{\rd b}{\rd t_k} = f_k, \quad 
  \frac{\rd c}{\rd t_k} = g_k, \quad 
  \frac{\rd b}{\rd\bar{t}_k} = \bar{f}_k, \quad 
  \frac{\rd c}{\rd\bar{t}_k} = \bar{g}_k. 
\label{bc-eeq}
\eeq
This implies that the factorized form (\ref{AL-LLbar}) 
of the Lax operators is preserved by the flows 
of the 2D Toda hierarchy.  
(\ref{bc-eeq}) can be further converted 
to the usual Hamiltonian form \cite{Suris97} 
of the Ablowitz-Ladik hierarchy 
by a change of variables \cite{BCR11}.

\subsubsection*{Another formulation of reduction to Ablowitz-Ladik hierarchy}

The Lax operators of (\ref{AL-LLbar}) have another expression 
\beq
  L = \tilde{C}^{-1}\tilde{B},\quad 
  \bar{L}^{-1} = \tilde{B}^{-1}\tilde{C}, 
\label{AL-LLbar-tilde}
\eeq
where $\tilde{B}$ and $\tilde{C}$ are difference operators 
of the form 
\beqnn
  \tilde{B} = e^{\rd_s} - \tilde{b},\quad 
  \tilde{C} = 1 - \tilde{c}e^{-\rd_s} 
\eeqnn
that satisfies the relation 
\beq
  \tilde{C}B = \tilde{B}C. 
\label{BC-BCtilde}
\eeq
$\tilde{B}^{-1}$ and $\tilde{C}^{-1}$ are understood 
to be defined by series expansions like (\ref{BC-inverse}).   
The operator relation (\ref{BC-BCtilde}) 
turns into the functional relations 
\beq
  b(s) + \tilde{c}(s) = \tilde{b}(s) + c(s+1), \quad 
  \tilde{c}(s)b(s-1) = \tilde{b}(s)c(s). 
\eeq
$b,c$ and $\tilde{b},\tilde{c}$ are thus mutually connected 
by an invertible transformation.  
The Lax equations (\ref{2DToda-Laxeq}) can be reduced 
to the equations 
\beq
\begin{gathered}
  \frac{\rd\tilde{B}}{\rd t_k} 
  = \left((\tilde{B}\tilde{C}^{-1})^k\right)_{\ge 0}\tilde{B} 
    - \tilde{B}\left(\tilde{C}^{-1}\tilde{B})^k\right)_{\ge 0},\\
  \frac{\rd\tilde{C}}{\rd t_k}
  = \left((\tilde{B}\tilde{C}^{-1})^k\right)_{\ge 0}\tilde{C} 
    - \tilde{C}\left((\tilde{C}^{-1}\tilde{B})^k\right)_{\ge 0},\\
  \frac{\rd\tilde{B}}{\rd\bar{t}_k}
  = \left((\tilde{C}\tilde{B}^{-1})^k\right)_{<0}\tilde{B} 
    - \tilde{B}\left((\tilde{B}^{-1}\tilde{C})^k\right)_{<0},\\
  \frac{\rd\tilde{C}}{\rd\bar{t}_k}
  = \left((\tilde{C}\tilde{B}^{-1})^k\right)_{<0}\tilde{C} 
    - \tilde{C}\left((\tilde{B}^{-1}\tilde{C})^k\right)_{<0} 
\end{gathered}
\label{BCtilde-Laxeq}
\eeq
for $\tilde{B}$ and $\tilde{C}$.  
These equations can be further converted 
to evolution equations for $\tilde{b}$ and $\tilde{c}$.  

The factorized form (\ref{AL-LLbar-tilde}) is suited 
for a comparison with the relativistic Toda hierarchy 
\cite{Ruijsenaars90}.  If the Lax operators are factorized 
in this form, the eigenvalue problem 
\beq
  L\Psi = z\Psi, \quad \bar{L}^{-1}\bar{\Psi} = z^{-1}\bar{\Psi}
\eeq
for the Baker-Akhiezer functions $\Psi$ and $\bar{\Psi}$ 
of the 2D Toda hierarchy \cite{UT84,TT95} can be converted 
to a ``generalized eigenvalue problem'' of the form 
\beq
  \tilde{B}\Psi = z\tilde{C}\Psi, \quad 
  \tilde{B}\bar{\Psi} = z\tilde{C}\bar{\Psi}. 
\eeq
It is this generalized eigenvalue problem 
that is used in Bruschi and Ragnisco's 
scalar-valued Lax formalism \cite{BR89} 
of the relativistic Toda lattice.  
As pointed out by Kharchev et al. \cite{KMZ96}, 
this generalized eigenvalue problem can be derived 
from the traditional $2\times 2$ matrix-valued Lax formalism 
\cite{AL75} of the Ablowitz-Ladik hierarchy as well.

\subsection{Factorization problem of infinite matrices}

Difference operators in the Lax formalism 
of the 2D Toda hierarchy are in one-to-one correspondence 
with $\ZZ\times\ZZ$ matrices by the rule 
\beqnn
  \sum_{k=-\infty}^\infty a_k(s)e^{k\rd_s} 
  \;\longleftrightarrow\; 
  \sum_{k=-\infty}^\infty \diag[a_k(s)]\Lambda^k, 
\eeqnn
where $\diag[a(s)]$ stands for the diagonal matrix 
\beqnn
  \diag[a(s)] = \sum_{i,j\in\ZZ}a(i)E_{ii}. 
\eeqnn
Since this correspondence preserves the algebraic structure, 
one can reformulate the Lax formalism of the 2D Toda hierarchy 
in terms of $\ZZ\times\ZZ$ matrices.   Let us use 
the same notations $L$, $\bar{L}$, $B_n$, $\bar{B}_n$ etc. 
for the corresponding $\ZZ\times\ZZ$ matrices.  

In this matrix formulation, one can formulate 
a matrix factorization problem that can capture 
all solutions of the 2D Toda hierarchy 
\cite{Takasaki84,NTT95,Takasaki95}.  
Given a constant invertible $\ZZ\times\ZZ$ matrix $U$, 
the problem is to find two $\ZZ\times\ZZ$ matrices 
$W = W(\bst,\bar{\bst})$ and 
$\bar{W} = \bar{W}(\bst,\bar{\bst})$ 
with the following properties: 
\begin{itemize}
\item[(i)] $W$ is a lower triangular matrix, 
and all diagonal elements are equal to $1$. 
\item[(ii)] $\bar{W}$ is an upper triangular matrix, 
and all diagonal elements are non-zero. 
\item[(iii)] $W$ and $\bar{W}$ satisfy the factorization relation
\beq
  \exp\left(\sum_{k=1}^\infty t_k\Lambda^k\right)
    U \exp\left(- \sum_{k=1}^\infty\bar{t}_k\Lambda^{-k}\right) 
  = W^{-1}\bar{W}. 
\label{MFP}
\eeq
\end{itemize}
If such matrices $W,\bar{W}$ exist, 
they satisfy the so called ``Sato equations'' 
\beq
\begin{aligned}
  \frac{\rd W}{\rd t_k} = B_kW - W\Lambda^k,\quad& 
  \frac{\rd\bar{W}}{\rd t_k} = B_k\bar{W},\\
  \frac{\rd W}{\rd\bar{t}_k} = \bar{B}_kW,\quad& 
  \frac{\rd\bar{W}}{\rd\bar{t}_k} 
    = \bar{B}_k\bar{W} - \bar{W}\Lambda^{-k}, 
\end{aligned}
\label{Sato-eq}
\eeq
where $B_k$ and $\bar{B}_k$ are constructed 
from $W$ and $\bar{W}$ as 
\beqnn
  B_k = \left(W\Lambda^kW^{-1}\right)_{\ge 0},\quad 
  \bar{B}_k = \left(\bar{W}\Lambda^{-k}\bar{W}^{-1}\right)_{<0}. 
\eeqnn
$(\quad)_{\ge 0}$ and $(\quad)_{<0}$
stand for the upper triangular part and 
the strictly lower triangular part.  
The Sato equations (\ref{Sato-eq}) imply 
that the matrices  $L$ and $\bar{L}$ defined by 
\beq
  L = W\Lambda W^{-1}, \quad 
  \bar{L} = \bar{W}\Lambda^{-1}\bar{W}^{-1} 
\label{LLbar-WWbar}
\eeq
satisfy the Lax equations (\ref{2DToda-Laxeq}).  

The factorization problem (\ref{MFP}) is directly connected 
with the fermionic expression \cite{Takebe91}
\beqnn
  \tau(s,\bst,\bar{\bst}) 
  = \langle s|\exp\left(\sum_{k=1}^\infty t_kJ_k\right)
    g \exp\left(- \sum_{k=1}^\infty\bar{t}_kJ_{-k}\right)
    |s\rangle
\eeqnn
of tau functions.  Namely, the $\ZZ\times\ZZ$ matrix 
$U = (u_{ij})$ corresponds to the operator $g$ 
by the Bogoliubov transformation 
\beq
  g\psi_jg^{-1} = \sum_{i\in\ZZ}\psi_iu_{ij} 
\eeq
of fermion creation-annihilation operators \cite{MJD-book}.  
Its infinitesimal form is exactly 
the correspondence (\ref{Xhat-def}) 
between $\ZZ\times\ZZ$ matrices and fermion bilinears. 

The operators (\ref{crystal-g}) and (\ref{crystal-g'}) 
defining the tau functions $\tau(s,\bst)$ and 
$\tau'(s,\bst,\bar{\bst})$ correspond to the matrices 
\beq
  U = q^{\Delta^2/2}\Gamma_{-}(q^{-\rho})\Gamma_{+}(q^{-\rho})Q^\Delta
      \Gamma_{-}(q^{-\rho})\Gamma_{+}(q^{-\rho})q^{\Delta^2/2}
\label{crystal-U}
\eeq
and 
\beq
  U' = q^{\Delta^2/2}\Gamma_{-}(q^{-\rho})\Gamma_{+}(q^{-\rho})
       Q^\Delta\Gamma'_{-}(q^{-\rho})\Gamma'_{+}(q^{-\rho})q^{-\Delta^2/2}, 
\label{crystal-U'}
\eeq
where $\Gamma_{\pm}(\bsx)$ and $\Gamma'_{\pm}(\bsx)$ 
are the matrix representation (\ref{GG'-matrix}) 
of the multi-variable vertex operators.  
We now consider the factorization problem (\ref{MFP}) 
in these particular cases.

\subsection{Initial values of Lax operators}

When $\bst$ and $\bar{\bst}$ are specialized to 
$\bst = \bar{\bst} = \bszero$, we can find 
an explicit form of the solution of 
the factorization problem (\ref{MFP}) 
for the matrix (\ref{crystal-U'}): 

\begin{lemma}
The factorization problem 
\beq
  U' = W^{-1}\bar{W}
\label{MFP-U'}
\eeq
for the matrix (\ref{crystal-U'}) can be solved as 
\beq
\begin{gathered}
  W = q^{\Delta^2/2}\Gamma'_{-}(Qq^{-\rho})^{-1}
      \Gamma_{-}(q^{-\rho})^{-1}q^{-\Delta^2/2},\\
  \bar{W} = q^{\Delta^2/2}Q^\Delta\Gamma_{+}(Qq^{-\rho})
            \Gamma'_{+}(q^{-\rho})q^{-\Delta^2/2}. 
\end{gathered}
\label{WWbar-U'}
\eeq
\end{lemma}

\begin{proof}
We can use the algebraic relations 
\beqnn
  Q^\Delta\Lambda^n Q^{-\Delta} = Q^{-n}\Lambda^n,\quad 
  Q^{-\Delta}\Lambda^n Q^\Delta = Q^n\Lambda^n 
\eeqnn
to rewrite the triple product in the middle 
of (\ref{crystal-U'}) as 
\begin{align*}
  \Gamma_{+}(q^{-\rho})Q^\Delta\Gamma'_{-}(q^{-\rho})
  &= Q^\Delta\Gamma_{+}(Qq^{-\rho})\Gamma'_{-}(q^{-\rho}) \nonumber\\
  &= Q^\Delta\Gamma'_{-}(q^{-\rho})\Gamma_{+}(Qq^{-\rho}) \nonumber\\
  &= \Gamma'_{-}(Qq^{-\rho})Q^\Delta\Gamma_{+}(Qq^{-\rho}). 
\end{align*}
Thus $U'$ becomes a product of triangular matrices: 
\beqnn
  U' = q^{\Delta^2/2}\Gamma_{-}(q^{-\rho})\Gamma'_{-}(Qq^{-\rho})
  \cdot Q^\Delta\Gamma_{+}(Qq^{-\rho})\Gamma'_{+}(q^{-\rho})q^{-\Delta^2/2}. 
\eeqnn
We have only to modify the first (lower triangular) part 
with the diagonal matrix $q^{-\Delta^2/2}$ so that 
all diagonal elements become $1$.  
\end{proof}

The initial values of the associated Lax operators 
(\ref{LLbar-WWbar}) at $\bst = \bar{\bst} = \bszero$ 
are obtained from (\ref{WWbar-U'}) as 
\beqnn
  L = W\Lambda W^{-1},\quad 
  \bar{L}^{-1} = \bar{W}\Lambda^{-1}\bar{W}^{-1}. 
\eeqnn
These matrices turn out to be a quotient of two matrices: 

\begin{lemma}
\beq
  L = (\Lambda - q^\Delta)(1 + Qq^{\Delta-1}\Lambda^{-1})^{-1}
    = (1 + Qq^\Delta\Lambda^{-1})^{-1}(\Lambda - q^\Delta). 
\label{L-factorized}
\eeq
\end{lemma}

\begin{proof}
By construction, $L$ is a matrix of the form 
\beqnn
  L = q^{\Delta^2/2}\Gamma'_{-}(Qq^{-\rho})^{-1}\Gamma_{-}(q^{-\rho})^{-1}
      q^{-\Delta^2/2}\Lambda q^{\Delta^2/2}
      \Gamma_{-}(q^{-\rho})\Gamma'_{-}(Qq^{-\rho})q^{-\Delta^2/2}. 
\eeqnn
Let us calculate this matrix step-by-step.  
Note that specialization of (\ref{SSiii-matrix}) 
to $k = 0$ reads 
\beq
  q^{\Delta^2/2}\Lambda^mq^{-\Delta^2/2} 
  = q^{m^2/2}\Lambda^mq^{-m\Delta}. 
\label{qDelta^2/2-Lambda-rel}
\eeq
This implies that 
\beqnn
  q^{-\Delta^2/2}\Lambda q^{\Delta^2/2} = \Lambda q^{\Delta-1/2}. 
\eeqnn
The next stage is to calculate the product 
\begin{align*}
  &\Gamma'_{-}(Qq^{-\rho})^{-1}\Gamma_{-}(q^{-\rho})^{-1}
      q^{-\Delta^2/2}\Lambda q^{\Delta^2/2}
      \Gamma_{-}(q^{-\rho})\Gamma'_{-}(Qq^{-\rho}) \nonumber\\
  &= \Lambda \Gamma'_{-}(Qq^{-\rho})^{-1}\Gamma_{-}
    (q^{-\rho})^{-1}q^{\Delta-1/2}\Gamma_{-}(q^{-\rho})
    \Gamma'_{-}(Qq^{-\rho}) 
\end{align*}
with the aid of (\ref{qDelta-Lambda-rel}).  
Firstly, the triple product of $\Gamma_{-}(q^{-\rho})^{\pm 1}$ 
and $q^\Delta$ can be calculated as follows: 
\begin{align*}
  \Gamma_{-}(q^{-\rho})^{-1}q^\Delta\Gamma_{-}(q^{-\rho}) 
  &= \Gamma_{-}(q^{-\rho})^{-1}\cdot 
     q^\Delta\Gamma_{-}(q^{-\rho})q^{-\Delta}\cdot q^\Delta 
  \nonumber\\
  &= \prod_{i=1}^\infty(1 - q^{i-1/2}\Lambda^{-1})\cdot
    \prod_{i=1}^\infty(1 - q^{i+1/2}\Lambda^{-1})^{-1}\cdot q^\Delta 
  \nonumber\\
  &= (1 - q^{1/2}\Lambda^{-1})q^\Delta. 
\end{align*}
Since $1 - q^{1/2}\Lambda^{-1}$ commutes 
with $\Gamma'_{-}(Qq^{-\rho})^{\pm 1}$, 
the next task is to calculate a triple product 
of $\Gamma'_{-}(Qq^{-\rho})^{\pm 1}$ and $q^\Delta$.  
This can be achieved as follows: 
\begin{align*}
  \Gamma'_{-}(Qq^{-\rho})^{-1}q^\Delta\Gamma'_{-}(Qq^{-\rho})
  &= \Gamma'_{-}(Qq^{-\rho})^{-1}\cdot 
     q^\Delta\Gamma'_{-}(Qq^{-\rho})q^{-\Delta}\cdot q^\Delta
  \nonumber\\
  &= \prod_{i=1}^\infty(1+Qq^{i-1/2}\Lambda^{-1})^{-1}\cdot
     \prod_{i=1}^\infty(1+Qq^{i+1/2}\Lambda^{-1})\cdot q^\Delta
  \nonumber\\
  &= (1+Qq^{1/2}\Lambda^{-1})^{-1}q^\Delta.
\end{align*}
These calculations yield the equality 
\beqnn
\begin{aligned}
  \Gamma'_{-}(Qq^{-\rho})^{-1}\Gamma_{-}(q^{-\rho})^{-1}
  q^{-\Delta^2/2}\Lambda q^{\Delta^2/2}
  \Gamma_{-}(q^{-\rho})\Gamma'_{-}(Qq^{-\rho})\\
  = \Lambda(1 - q^{1/2}\Lambda^{-1})
    (1+Qq^{1/2}\Lambda^{-1})^{-1} q^{\Delta-1/2}, 
\end{aligned}
\eeqnn
hence an expression of $L$ of the form 
\beqnn
  L = q^{\Delta^2/2}(\Lambda - q^{1/2})
      (1+Qq^{1/2}\Lambda^{-1})q^{\Delta-1/2} q^{-\Delta^2/2}. 
\eeqnn
We now use the algebraic relations 
\beqnn
\begin{gathered}
  q^{\Delta^2/2}\Lambda q^{-\Delta^2/2} = q^{-\Delta-1/2}\Lambda,\quad 
  q^{\Delta^2/2}\Lambda^{-1}q^{-\Delta^2/2} = q^{\Delta-1/2}\Lambda^{-1},\\
  \Lambda q^\Delta = q^{\Delta+1}\Lambda,\quad 
  \Lambda^{-1}q^\Delta = q^{\Delta-1}\Lambda^{-1} 
\end{gathered}
\eeqnn
that follow from (\ref{qDelta^2/2-Lambda-rel}) and 
(\ref{qDelta-Lambda-rel}) to rewrite $L$ further 
as follows: 
\begin{align*}
  L &= (q^{-\Delta-1/2}\Lambda - q^{1/2})
       (1 + Qq^\Delta\Lambda^{-1})q^{\Delta-1/2} \nonumber\\
  &= q^{\Delta-1/2}(q^{-\Delta+1/2}\Lambda - q^{1/2})
  (1 + Qq^{\Delta-1}\Lambda^{-1})^{-1} \nonumber\\
  &= (\Lambda - q^\Delta)(1 + Qq^{\Delta-1}\Lambda^{-1})^{-1}. 
\end{align*}
This gives the first part of (\ref{L-factorized}). 
In the same way, starting from 
\beqnn
  L = q^{\Delta^2/2}(1 + Qq^{1/2}\Lambda^{-1})^{-1} 
      (\Lambda - q^{1/2})q^{\Delta-1/2}q^{-\Delta^2/2}
\eeqnn
and repeating similar calculations, one can derive 
the second part of (\ref{L-factorized}). 
\end{proof}

\begin{lemma}
\beq
  \bar{L}^{-1}
  = (1 + Qq^{\Delta-1}\Lambda^{-1})(q^\Delta - \Lambda)^{-1}
  =  (q^\Delta - \Lambda)^{-1}(1 + Qq^\Delta\Lambda^{-1}). 
\label{Lbar-factorized}
\eeq
\end{lemma}

\begin{proof}
The proof is parallel to the case of the previous lemma. 
By construction, $\bar{L}^{-1}$ is a matrix of the form 
\beqnn
\begin{aligned}
  \bar{L}^{-1} 
  = q^{\Delta^2/2}Q^\Delta\Gamma_{+}(Qq^{-\rho})
    \Gamma'_{+}(q^{-\rho})q^{-\Delta^2/2}\Lambda^{-1}q^{\Delta^2/2}\\
  \mbox{}\times\Gamma'_{+}(q^{-\rho})^{-1}\Gamma_{+}(Qq^{-\rho})^{-1}
  Q^{-\Delta}q^{-\Delta^2/2}. 
\end{aligned}
\eeqnn
With the aid of (\ref{qDelta^2/2-Lambda-rel}) and 
(\ref{qDelta-Lambda-rel}), one can derive the equality 
\beqnn
\begin{gathered}
  \Gamma_{+}(Qq^{-\rho})\Gamma'_{+}(q^{-\rho})
  q^{-\Delta^2/2}\Lambda^{-1}q^{\Delta^2/2}
  \Gamma'_{+}(q^{-\rho})^{-1}\Gamma_{+}(Qq^{-\rho})^{-1} \\
  = \Lambda^{-1}(1 + q^{1/2}\Lambda)
    (1 - Qq^{1/2}\Lambda)^{-1}q^{-\Delta-1/2}. 
\end{gathered}
\eeqnn
This implies that 
\beqnn
  \bar{L}^{-1} 
  = q^{\Delta^2/2}Q^\Delta (1 + q^{-1/2}\Lambda^{-1})
    (1 - Qq^{1/2}\Lambda)^{-1}q^{-\Delta}Q^{-\Delta}q^{-\Delta^2/2}.  
\eeqnn
It is easy to derive the final form (\ref{Lbar-factorized}) 
of $\bar{L}^{-1}$ from this expression. 
\end{proof}

\subsection{Identification of solution}

Translated to the language of difference operators, 
the foregoing results (\ref{L-factorized}) 
and (\ref{Lbar-factorized}) show that 
the initial values of the Lax operators 
at $\bst = \bar{\bst} = \bszero$ are factorized as 
\beq
  L = BC^{-1} = \tilde{C}^{-1}\tilde{B}, \quad 
  \bar{L}^{-1} = - CB^{-1} = - \tilde{B}^{-1}\tilde{C} 
\label{AL-LLbar-minus}
\eeq
with factors of the special form 
\beq
  B = \tilde{B} = e^{\rd_s} - q^s, \quad 
  C = 1 + Qq^{s-1}e^{-\rd_s},\quad
  \tilde{C} = 1 + Qq^s e^{-\rd_s}.
\label{BC-initial}
\eeq
Note that this expression is slightly different 
from (\ref{AL-LLbar}) and (\ref{AL-LLbar-tilde}). 
Namely, $\bar{L}^{-1}$ has an extra negative sign.
This does not affect the essential part 
of the reduction procedure in the case of (\ref{AL-LLbar}).   
The outcome of this procedure is the ``twisted''version 
\beq
  \frac{\rd b}{\rd t_k} = f_k, \quad 
  \frac{\rd c}{\rd t_k} = g_k, \quad 
  (-1)^k\frac{\rd b}{\rd\bar{t}_k} = \bar{f}_k, \quad 
  (-1)^k\frac{\rd c}{\rd\bar{t}_k} = \bar{g}_k 
\label{bc-evolutioneq}
\eeq
of the evolution equations (\ref{bc-eeq}).  
The reduced system is the same Ablowitz-Ladik hierarchy 
except that the second set of time variables $\bar{\bst}$ 
are replaced by $\iota(\bar{\bst}) 
= (-t_1,t_2,-t_3,\ldots)$.  

Thus we can conclude that the factorized Lax operators 
(\ref{AL-LLbar-minus}) with the initial data (\ref{BC-initial}) 
at $\bst=\bar{\bst}=\bszero$ persist to be factorized 
after time evolutions.  In summary, we have proven 
the following proposition: 

\begin{theorem}
The factorization problem (\ref{MFP}) 
for the matrix (\ref{crystal-U'}) yields 
a solution of the 2D Toda hierarchy 
with the tau function (\ref{crystal-tau'}). 
The associated Lax operators $L$ and $\bar{L}$ 
are factorized in the form of (\ref{AL-LLbar-minus}). 
In particular, this is a solution 
of the Ablowitz-Ladik hierarchy with 
the second set of time variables $\bar{\bst}$ 
being replaced by $\iota(\bar{\bst})$.  
\end{theorem}

\begin{remark}
Unlike the reduction condition (\ref{AL-LLbar}) 
of Brini et al. \cite{Brini10,BCR11}, 
our reduction condition (\ref{AL-LLbar-minus}) 
has extra negative signs.  This apparent discrepancy 
seems to be related to the fact that the usual expression 
\cite{AKMV03,CGMPS06} of the string amplitude $Z_{X_1}$ 
of the resolved conifold contains the parameter $Q$ 
with with a negative sign: 
\beq
  Z_{X_1} = \sum_{\lambda\in\calP}
    s_\lambda(q^{-\rho})s_{\tp{\lambda}}(q^{-\rho})(-Q)^{|\lambda|}. 
\eeq
\end{remark} 

\begin{remark}
The factorization problem (\ref{MFP}) for the matrix 
(\ref{crystal-U}), too, can be treated in the same way.  
The solution at $\bst = \bar{\bst} = \bszero$ reads 
\beq
\begin{gathered}
  W = q^{\Delta^2/2}\Gamma_{-}(Qq^{-\rho})^{-1}
      \Gamma_{-}(q^{-\rho})^{-1}q^{-\Delta^2/2},\\
  \bar{W} = q^{\Delta^2/2}Q^\Delta\Gamma_{+}(Qq^{-\rho})
      \Gamma_{+}(q^{-\rho})q^{\Delta^2/2}.
\end{gathered}
\eeq
The initial values of the associated Lax operators 
at $\bst = \bar{\bst} = \bszero$ 
have the factorized form 
\beq
  L = \bar{L}^{-1} 
  = e^{\rd_s}(1 - q^{s-1}e^{-\rd_s})(1 - Qq^{s-1}e^{-\rd_s}). 
\eeq
Such a factorized form is preserved by the flows 
of the 2D Toda hierarchy.  
\end{remark}

\subsection*{Acknowledgements}

This work is partly supported by JSPS Grants-in-Aid 
for Scientific Research No. 21540218 and No. 22540186 
from the Japan Society for the Promotion of Science.

\end{document}